\documentclass[12pt]{iopart}

\expandafter\let\csname equation*\endcsname\relax
\expandafter\let\csname endequation*\endcsname\relax 

\usepackage{amsmath}
\usepackage{amssymb}
\usepackage{epsfig}
\usepackage{color}

\graphicspath{{.}{./EPS/}}

\newcommand* {\ee}{\ensuremath{\mathrm{e}}}

\newtheorem{theorem}{Theorem}

\newtheorem{proposition}[theorem]{Proposition}

\newenvironment{proof}[1][Proof]{\begin{trivlist}
\item[\hskip \labelsep {\bfseries #1}]}{\end{trivlist}}

\newenvironment{example}[1][Example]{\begin{trivlist}
\item[\hskip \labelsep {\bfseries #1}]}{\end{trivlist}}

\newenvironment{remark}[1][Remark]{\begin{trivlist}
\item[\hskip \labelsep {\bfseries #1}]}{\end{trivlist}}

\newcommand{\qed}{\nobreak \ifvmode \relax \else
      \ifdim\lastskip<1.5em \hskip-\lastskip
      \hskip1.5em plus0em minus0.5em \fi \nobreak
      \vrule height0.75em width0.5em depth0.25em\fi}

\begin{document}
 
\title[]{Dynamical control of quantum systems in the context of mean ergodic theorems}
\author{J. Z. Bern\'ad}
\address{Institut f\"{u}r Angewandte Physik, Technische Universit\"{a}t Darmstadt, D-64289 Darmstadt, Germany}
\ead{Zsolt.Bernad@physik.tu-darmstadt.de}

\date{\today}

\begin{abstract}
Equidistant and non-equidistant single pulse "bang-bang" dynamical controls are investigated in the context of mean ergodic theorems. We 
show the requirements in which the limit of infinite pulse control for both the equidistant and the non-equidistant dynamical control converges to the same 
unitary evolution. It is demonstrated that the generator of this evolution can be obtained by projecting the generator of the free evolution 
onto the commutant of the unitary operator representing the pulse. Inequalities are derived to prove this statement and in the case of non-equidistant approach these 
inequalities are optimised as a function of the time intervals. 
\end{abstract}

\vspace{2pc}

\maketitle

\section{Introduction}
\label{I}

One of the basic requirements of quantum information processing is the reliability of the physical qubits \cite{DiVincenzo}. A possible method to deal
with this problem is the so-called dynamical decoupling which allows the suppression of unwanted environmental effects \cite{Viola1,Viola2,Zanardi,Vitali}. 
Its source of motivation lies in the spin-echo effect \cite{hahn} and the nuclear magnetic spectroscopy community have already developed various decoupling 
methods to eliminate the dephasing of the spins \cite{Carr,Meiboom,Haeberlen}. The formalism proposed by L. Viola and S. Lloyd in Ref. \cite{Viola1} is mainly
based on a "bang-bang" control where unitary pulses are applied instantaneously and equidistantly separated in time to the quantum system in order to cancel undesirable parts of 
the Hamiltonian evolution. In recent years dynamical decoupling has been scrutinized more closely both from the theoretical \cite{Viola3,Uhrig,Khodjatesh,Quiroz} 
and the experimental side \cite{Biercuk,Bluhm,Souza,Hayes,Piltz}. It has been shown that this method can also tailor the Hamiltonian evolution into a 
desired one \cite{Escher,Frydrych}. Therefore, it can be called dynamical control and not only dynamical decoupling.

The simplest problem can be formulated as the quest for the following limit
\begin{equation}
 \lim_{N \to \infty} \left(u \ee^{-i H t/N} \right)^N=\ee^{-i H_{\text{id}}t}
 \label{1}
\end{equation}
where $H$ is the Hamiltonian operator of the evolving system, $H_{\text{id}}$ is the desired Hamilton operator and $u$ is the instantaneously applied unitary operator. 
This question falls into the collection of semigroup product formulas and it has been also connected to the quantum Zeno effect \cite{Arendt,Facchi}. The Chernoff product formula 
\cite{Chernoff} can be applied provided that there exists an non-zero natural number $k$ such that $u^k$  is equal to the identity \cite{Hillier}. 
This approach results in the formula of averaging the Hamilton operator $H$ over the group $\{u,u^2,...,u^k\}$ \cite{Zanardi}.
Another possibility is to determine the generator of the unitary operator in \eqref{1} for a fixed $N$ and then to
study the limit $N \to \infty$ of the generator series with the help of the von Neumann's mean ergodic theorem \cite{Facchi}. The latter method assumes that
the unitary operator $u$ has only a non-degenerate point spectrum and shows that in the case of infinite pulses we obtain a unitary evolution which is governed by
the Hamilton operator $H$ projected onto the commutant of $u$. Both methods can be applied to unbounded Hamilton operators with well defined domain.

In the present paper we have three aims: first, to extend the discussion of the limit in Eq. \eqref{1} to unitary operators with arbitrary spectrum; second, to generalise
this equation towards the non-equidistant dynamical control case \cite{Uhrig}; third, to investigate to some extent the optimisation of the 
convergence in the latter. Our
work attempts to take the full advantage of the results in ergodic theory \cite{Krengel}. In the context of the dynamical control we will work 
with the Banach space of bounded 
linear operators on a Hilbert space. We make use of the ergodic theorems 
obtained on Banach spaces and investigate the Ces\`{a}ro mean $N^{-1} \sum^{N-1}_{i=0} T^ix$, where $x$ is an element of the Banach space and 
$T$ is a bounded linear operator. 
Furthermore we generalise the problem in  Eq. \eqref{1} such that the system's unitary evolution can also be replaced by an 
uniformly continuous one-parameter semigroup of operators. 
We prove the convergence of the generalised product formula with the help of mean ergodic theorems. We will show that weighted Ces\`{a}ro means are 
directly connected to 
non-equidistant dynamical control. The abstract mean ergodic theorems by W. F. Eberlein already cover both standard and 
weighted Ces\`{a}ro means in locally-convex linear topological spaces \cite{Eberlein} with the help of the Mazur-Bourgin theorem \cite{Bourgin}. 
The four equivalent statements of Eberlain's theorem are too abstract to be applied directly to the problem formulated in this work. For this reason 
it is more convenient to use the splitting theorems of K. Yosida \cite{Yosida} and L. W. Cohen \cite{Cohen} for a Banach space. These approaches apply also
to the closed linear subspace of a Banach space, where the limit of the Ces\`{a}ro means exist and which will be the case in our investigation. 
The quests in these theorems is to show that the sequence of Ces\`{a}ro means converges strongly if a 
subsequence converges weakly, and by thus they are contained in Eberlein's abstract theorem. However, their formulation is favourable for our task and 
furthermore Cohen's theorem defines the sufficient conditions of the weights in a weighted Ces\`{a}ro mean, which imply the strong convergence.

The paper is organized as follows. In section \ref{II} we connect the equidistant dynamical control with the mean ergodic theorem of K. Yosida. The main theorem clarifies the convergence
of Eq. \eqref{1} and shows that the generator of the evolution in the limit $N \to \infty$ is obtained by projecting the original generator onto the commutant of $u$. In section 
\ref{III} we use the same strategy as in section \ref{II} for non-equidistant dynamical control. The main theorem shows under which conditions a limit can be obtained and if it exits then 
it is the same as in the case of equidistant dynamical control. The upper bounds of the inequalities derived in the main theorem are functions of the weights related to the 
non-equidistant splitting of the time and we give an optimisation for these functions.

\section{Equidistant dynamical control}
\label{II}

The main goal of dynamical control is to take active control over the time evolution of a system and change it
to a desirable way. This usually means that the there is a Hamilton operator $H$ which governs the free evolution
and we would like to change it to $H_{\text{id}}$. The source of inspiration lies in the method of dynamical decoupling
where the aim is to decouple two interacting systems($A$ and $B$), i.e.,
\begin{eqnarray}
H=H_A \otimes I_B+ I_A \otimes H_B +H_{AB}, \quad
H_{\text{id}}=H_A \otimes I_B+ I_A \otimes H_B, \nonumber
\end{eqnarray}
with $H_A$($H_B$) being the Hamilton operator of system $A$($B$), $I_A$($I_B$) is the identity operator in system  $A$($B$) and
$H_{AB}$ is the Hamiltonian operator of the interaction. In order to achieve active control over the time evolution it is assumed that the available time $t$
of the evolution is divided into intervals of length $\tau =t/N$ and an instantaneous unitary pulse $u$ is carried out after each time interval $\tau$. 
Therefore the resulting time evolution after applying all $N$ pulses is governed by the unitary operator
\begin{equation}
 u_N(t)=u \ee^{-iHt/N} u \ee^{-iHt/N} \dots u \ee^{-iHt/N}.
 \label{Ueq}
\end{equation}
The question is that of determining $u$ and the value of $N$ such that $u_N(t)$ gets close to $\ee^{-iH_{\text{id}}t}$ in an appropriately chosen norm.
In order to answer this question we are going to take the following approach: first we determine $u$ in the limit $N \to \infty$; second we derive an $N$ dependent 
upper bound for the distance between $u_N(t)$ and $\ee^{-iH_{\text{id}}t}$.

Let ${\cal B}(\mathcal{H})$ be the set of all bounded linear operators on a Hilbert space $\mathcal{H}$.  ${\cal B}(\mathcal{H})$ is a Banach space with respect to the operator norm
\begin{equation}
 ||A||_{\text{op}}=\sup\{||Ax||: x \in \mathcal{H}, ||x|| \leqslant 1\}. \nonumber
\end{equation}
We consider that $H \in {\cal B}(\mathcal{H})$. Therefore, we can express $\ee^{-iHt/N}$ by its infinite series and substituting in Eq. \eqref{Ueq} we obtain
\begin{eqnarray}
 u_N(t)= u^N-it\left(\frac{1}{N}\sum^N_{k=1} u^k H (u^\dagger)^k\right)u^N
 - \frac{t^2}{2 N^2} \left(\sum^N_{k=1} u^k H^2 (u^\dagger)^k+ \dots \right) u^N, \nonumber
\end{eqnarray}
where $u^\dagger$ is the adjoint of $u$.  

At first sight it seems that it is demanding to deal with the above expansion, but in the following we demonstrate step by step that this formula 
is deeply connected to ergodic theorems and the convergence of every order can be evaluated. We notice that the second term is a Ces\`{a}ro mean:
\begin{equation}
\frac{1}{N}\sum^N_{k=1} u^k H (u^\dagger)^k=\frac{1}{N}\sum^N_{k=1} T^k(H),
\label{eqmean}
\end{equation}
where $T$ is a linear operator on the Banach space ${\cal B}(\mathcal{H})$. Let us start with the following simple statement.
\begin{proposition}
\label{prop1}
Let $u$ be a unitary operator on a Hilbert space $\mathcal{H}$ and $T:{\cal B}(\mathcal{H})\rightarrow {\cal B}(\mathcal{H})$. 
If $T(A)=uAu^\dagger$ for all $A \in {\cal B}(\mathcal{H})$ then $T$ is an isometry.
\end{proposition}
\begin{proof}
An elementary property of ${\cal B}(\mathcal{H})$ is that for all $A \in {\cal B}(\mathcal{H})$, we have
\begin{equation}
 ||A||^2_{\text{op}}=||A^\dagger A||_{\text{op}}. \nonumber
\end{equation}
On the other hand for any two unitary operators $u_1$, $u_2$ 
\begin{eqnarray}
 ||Au_2||^2_{\text{op}}=||u^\dagger_2 A^\dagger u^\dagger_1 u_1 A u_2||_{\text{op}}
=|| \left(u_1 A u_2 \right)^\dagger u_1 A u_2||_{\text{op}}=||u_1Au_2||^2_{\text{op}} \nonumber
\end{eqnarray}
and the equality
\begin{equation}
 ||Au_2||_{\text{op}}=\sup\limits_{||x|| \leqslant 1}||Au_2x||=\sup\limits_{||x'||  \leqslant 1}||Ax'||=||A||_{\text{op}} \nonumber
\end{equation}
shows that the operator norm is a unitarily invariant norm. It is immediate that $||T(A)||_{\text{op}}=||A||_{\text{op}}$ for all $A \in {\cal B}(\mathcal{H})$
and therefore $T$ is an isometry.
\end{proof}

\begin{remark}
In the case of Hilbert-Schmidt operators which form a Hilbert-space
\begin{equation}
{\cal B}_2(\mathcal{H})=\{X \in {\cal B}(\mathcal{H}): ||X||_2=\sqrt{\sum_{i \in I}||Xe_i||^2}<\infty \}, \nonumber
\end{equation}
where $(e_i)_{i \in I}$ is an orthonormal basis of $\mathcal{H}$ the linear operator $T$ is unitary. This can be shown by using the Hilbert-Schmidt inner product
\begin{eqnarray}
\langle A,T(B) \rangle&=&\mathrm{Tr}\{ A^\dagger u B u^\dagger\}=\mathrm{Tr}\{u^\dagger A^\dagger u B\}=\langle T^\dagger (A),B \rangle,\,\,\forall A,B \in {\cal B}_2(\mathcal{H}), \nonumber \\
\mathrm{Tr}\{A\}&=&\sum_{i \in I} \langle e_i, A e_i \rangle_{\mathcal{H}} \nonumber
\end{eqnarray}
and the inner product $\langle .\,, .\rangle_{\mathcal{H}}$ of $\mathcal{H}$ to define the adjoint map $T^\dagger(A)=u^\dagger A u$ which obeys
\begin{equation}
 T T^\dagger=T^\dagger T= \mathcal{I} \nonumber
\end{equation}
with $\mathcal{I}$ being the identity map on ${\cal B}_2(\mathcal{H})$.
\end{remark}

Let us consider the operator sequence
\begin{equation}
 T_N (X)=\frac{1}{N}\sum^N_{k=1} T^k(X)
 \label{eqdynseq}
\end{equation}
in ${\cal B}(\mathcal{H})$. It is immediate from Proposition \ref{prop1} that
\begin{equation}
||T_N (X)||_{\text{op}}\leqslant||X||_{\text{op}} \nonumber
\end{equation}
and $T$ maps the closed unit ball $\{X \in {\cal B}(\mathcal{H}): ||X||_{\text{op}} \leqslant 1 \}$ to itself. We recall form operator theory that the dual 
of the trace class operators
\begin{equation}
   {\cal B}_1(\mathcal{H}):=\{ X \in {\cal B}(\mathcal{H}): \mathrm{Tr}{\sqrt{ X^\dagger X}} <\infty\}. \nonumber
\end{equation}
is ${\cal B}(\mathcal{H})$. Therefore the ultraweak operator topology is just the weak$^*$ topology on ${\cal B}(\mathcal{H})$ and by the 
Banach-Alaoglu theorem the closed unit ball of ${\cal B}(\mathcal{H})$ is compact in the ultraweak operator topology. 
The ultraweak operator topology is weaker (coarser) than the weak Banach topology and ${\cal B}(\mathcal{H})$ in general is not 
a reflexive space. Therefore, the closed unit ball in ${\cal B}(\mathcal{H})$ in general is not weakly compact \cite{Reed, Murphy}. 

These arguments show that $T$ is a power bounded operator on ${\cal B}(\mathcal{H})$:
\begin{equation}
\sup_k\{||T^k(X)||_{\text{op}}: X \in {\cal B}(\mathcal{H}), ||X||_{\text{op}} \leqslant 1\}=1,  \nonumber 
\end{equation}
but is not weakly compact. We introduce the following set for a power bounded linear operator $T$ on a Banach space ${\cal B}$ 
\begin{equation}
 \Xi_{{\cal B}}=\{x \in {\cal B}: \lim \frac{1}{N}\sum^N_{k=1} T^kx \,\, \text{exists}\}, \nonumber
\end{equation}
which is a closed linear subspace of ${\cal B}$ \cite{Krengel}. We define also the linear subspace
\begin{equation}
F_{{\cal B}}=\{x \in {\cal B}: Tx=x\}. \nonumber
\end{equation}

The original version of K. Yosida's mean ergodic theorem assumes that the operator $T$ in the 
Ces\`{a}ro mean is weakly compact. This is not the case in our work as it has been shown above. 
Therefore, we shall use a slightly modified version of K. Yosida's mean ergodic theorem, which states \cite{Krengel,Yosida}:
\begin{theorem}[Yosida's ergodic theorem]
 \label{theorem1}
 Let $T$ be a power bounded linear operator on a Banach space ${\cal B}$. Then 
 \begin{equation}
  \Xi_{{\cal B}}=F_{{\cal B}} \oplus \overline{\{x-Tx: x \in {\cal B}\}}. \nonumber
 \end{equation}
The linear operator $Px=\lim \frac{1}{N}\sum^N_{k=1} T^kx$ for $x \in \Xi_{{\cal B}}$ is the projection of $\Xi_{{\cal B}}$ onto $F_{{\cal B}}$.
We have $P=P^2=TP=PT$ and for any $z\in {\cal B}$ the assertions
\begin{eqnarray}
 \text{1)}\,\,\, \lim \frac{1}{N}\sum^N_{k=1} T^kz=0, \nonumber \\
 \text{2)}\,\,\,  z \in \overline{\{x-Tx: x \in {\cal B}\}} \nonumber 
\end{eqnarray}
are equivalent.
\end{theorem}

Applying this theorem to the sequence \eqref{eqdynseq} with $T(X)=u X u^\dagger$ being a power bounded linear operator we find that
\begin{eqnarray}
 &&\lim_{N \to \infty}\frac{1}{N}\sum^N_{k=1} u^k X (u^\dagger)^k=P(X),\, \forall X \in  \Xi_{{\cal B}(\mathcal{H})}, \nonumber \\
 &&\Xi_{{\cal B}(\mathcal{H})}=F_{{\cal B}(\mathcal{H})}\oplus \overline{\{X-uXu^\dagger: X \in {\cal B}(\mathcal{H})\}} \nonumber
\end{eqnarray}
where $P$ projects onto the linear subspace $F_{{\cal B}(\mathcal{H})}=\{X \in {\cal B}(\mathcal{H}): [u,X]=0\}$.  
If $X^\dagger=X \in \Xi_{{\cal B}(\mathcal{H})}$, then ${\cal P}(X)^\dagger={\cal P}(X)$.

Now, considering all these preparations we return to the operator series given in Eq. \eqref{Ueq} and we give the main result of this section.
\begin{theorem}
 \label{theorem2}
 Let $u$ be a unitary operator in the Hilbert space $\mathcal{H}$ and the set $\Xi_{{\cal B}(\mathcal{H})}$ is defined
 by the power bounded linear map $T(X)=u X u^\dagger$. Let $P$ be the projection 
 operator which maps $\Xi_{{\cal B}(\mathcal{H})}$ onto the linear subspace $F_{{\cal B}(\mathcal{H})}=\{X \in {\cal B}(\mathcal{H}): [u,X]=0\}$. 
 Then, for any $X \in \Xi_{{\cal B}(\mathcal{H})}$ 
 and $t \in \mathbb{C}$ with $|t|<\infty$
 \begin{equation}
  \lim_{N \to \infty} ||u \ee^{Xt/N} u \ee^{Xt/N} \dots u \ee^{Xt/N}-\ee^{P(X)t} u^N||_{\text{op}} = 0.
\label{limit1}
 \end{equation}
\end{theorem}
\begin{proof}
 The equality 
\begin{equation}
 u \ee^{Xt/N}  \dots u \ee^{Xt/N}=\ee^{uXu^\dagger t/N} \dots  \ee^{u^N X (u^\dagger)^N t/N} u^N \nonumber
\end{equation}
combined with the unitarily invariant property of the operator norm results
\begin{eqnarray}
 &&||u \ee^{Xt/N} u \ee^{Xt/N} \dots u \ee^{Xt/N}-\ee^{P(X)t} u^N||_{\text{op}}= \nonumber \\
&&=||\ee^{uXu^\dagger t/N} \ee^{u^2 X (u^\dagger)^2 t/N} \dots  \ee^{u^N X (u^\dagger)^N t/N}-\ee^{P(X)t}||_{\text{op}}. \nonumber
\end{eqnarray}
An outline to the strategy of the proof is the following: in the first part we are going to prove the convergence for elements 
either in $F_{{\cal B}(\mathcal{H})}$ or $\{X-uXu^\dagger: X \in {\cal B}(\mathcal{H})\}$; in the second part we make use 
the results of the first part and prove the convergence for any $X \in \Xi_{{\cal B}(\mathcal{H})}$.  In both cases systematic approximants of the 
product formula
\begin{equation}
 \ee^{uXu^\dagger t/N} \ee^{u^2 X (u^\dagger)^2 t/N} \dots  \ee^{u^N X (u^\dagger)^N t/N} \nonumber
\end{equation}
are studied.

First let $X=Y-uYu^\dagger$, that is, $X \in \{Y-uYu^\dagger: Y \in {\cal B}(\mathcal{H})\}$ and we define the following operator:
\begin{equation}
 S_n=\ee^{\sum^n_{i=1}u^i X (u^\dagger)^i t/N} \prod^N_{j=n+1} \ee^{u^j X (u^\dagger)^j t/N},\, 1\leqslant n < N. \nonumber
\end{equation}
Then
\begin{eqnarray}
||\ee^{uXu^\dagger t/N} \ee^{u^2 X (u^\dagger)^2 t/N} \dots  \ee^{u^N X (u^\dagger)^N t/N}-\ee^{P(X)t}||_{\text{op}}=||S_1-\ee^{P(X)t}||_{\text{op}}\leqslant \label{start}\\
 \leqslant \sum^{N-2}_{n=1} ||S_n-S_{n+1}||_{\text{op}}+||S_{N-1}-\ee^{\sum^N_{i=1}u^i X (u^\dagger)^i t/N}||_{\text{op}}
+ ||\ee^{\sum^N_{i=1}u^i X (u^\dagger)^i t/N}-\ee^{P(X)t}||_{\text{op}}. \nonumber
\end{eqnarray}
It follows from the submultiplicative property of $||.||_{\text{op}}$  that
\begin{eqnarray}
 &&||S_n-S_{n+1}||_{\text{op}}\leqslant  \label{ineq1} \\
&&\leqslant \prod^N_{j=n+2} ||\ee^{u^j X (u^\dagger)^j t/N}||_{\text{op}}||\ee^{\sum^n_{i=1}u^i X (u^\dagger)^i t/N}\ee^{u^{n+1} X (u^\dagger)^{n+1} t/N}-
\ee^{\sum^{n+1}_{i=1}u^i X (u^\dagger)^i t/N}||_{\text{op}} \nonumber
\end{eqnarray}
for $N-1>n\geqslant1$. The exponential of a bounded operator $\ee^A$ is defined 
through its Taylor series
\begin{equation}
\ee^A=1+A+\frac{A^2}{2!}+\dots \nonumber 
\end{equation}
and therefore for all $A,B \in {\cal B}(\mathcal{H})$ we get
\begin{eqnarray}
 ||\ee^{A} \ee^{B}-\ee^{A+B}||_{\text{op}}\leqslant \sum^\infty_{i=2} \frac{2}{i!} \left[\sum^{i-1}_{j=1} 
 \left( {i \choose j}-1 \right) ||A||^j_{\text{op}} ||B||^{i-j}_{\text{op}}  \right]. \label{ineq2}
\end{eqnarray}
This inequality can be combined with \eqref{ineq1} by choosing 
\begin{eqnarray}
 A&=&\frac{1}{N}\sum^n_{i=1}u^i X (u^\dagger)^i\, t, \nonumber \\
 B&=&\frac{1}{N}u^{n+1} X (u^\dagger)^{n+1}\, t. \nonumber
\end{eqnarray}
It follows from $X=Y-uYu^\dagger$ and the unitarily invariant property of $||.||_{\text{op}}$  that
\begin{eqnarray}
\label{ineqYeq}
||\frac{1}{N}\sum^n_{i=1}u^i X (u^\dagger)^i ||_{\text{op}}&=&||\frac{1}{N}\sum^n_{i=1}u^i (Y-uYu^\dagger) (u^\dagger)^i ||_{\text{op}} \nonumber \\
&=&||\frac{uYu^\dagger-u^{n+1}Y(u^\dagger)^{n+1}}{N}||_{\text{op}}\leqslant \frac{2||Y||_{\text{op}}}{N},
\end{eqnarray}
and
\begin{eqnarray}
 &&||\sum^n_{i=1}u^i X (u^\dagger)^i \,t||_{\text{op}}\leqslant 2||Y||_{\text{op}}\,|t|,\quad 
 || u^{n+1} X (u^\dagger)^{n+1}\, t||_{\text{op}}\leqslant 2||Y||_{\text{op}}\, |t|, \nonumber \\
 &&||\ee^{u^j X (u^\dagger)^j t/N}||_{\text{op}}\leqslant \sum^\infty_{n=0} \frac{||u^j X (u^\dagger)^j t/N||^n_{\text{op}}}{n!} \leqslant \sum^\infty_{n=0} 
 \frac{2^n |t|^n ||Y/N||^n_{\text{op}}}{n!}=\ee^{\frac{2|t|}{N}||Y||_{\text{op}}}. \nonumber
\end{eqnarray}
Finally, the inequality in \eqref{ineq1} yields
\begin{eqnarray}
 ||S_n-S_{n+1}||_{\text{op}}&\leqslant& \ee^{2|t|\frac{N-n-1}{N} ||Y||_{\text{op}}} \sum^\infty_{i=2} \frac{2^{i+1}|t|^i}{i!N^i} ||Y||^i_{\text{op}} \sum^{i-1}_{j=1} 
 \left( {i \choose j}-1 \right) \nonumber \\
 &\leqslant& \ee^{2|t| ||Y||_{\text{op}}} \sum^\infty_{i=2} \frac{2^{i+1}|t|^i}{i!N^i} ||Y||^i_{\text{op}} \sum^{i-1}_{j=1} 
 \left( {i \choose j}-1 \right) \nonumber
\end{eqnarray}
and therefore
\begin{eqnarray}
\sum^{N-2}_{n=1} ||S_n-S_{n+1}||_{\text{op}} \leqslant (N-2) \ee^{2|t|\,||Y||_{\text{op}}}\left(\frac{4 |t|^2}{N^2}||Y||^2_{\text{op}}+
\mathcal{O}\left(\frac{1}{N^3}\right) \right), \label{1use}
\end{eqnarray}
where $\mathcal{O}$ is the big $O$ notation for asymptotic behaviour. We make use again of inequality \eqref{ineq2} by choosing
\begin{eqnarray}
 A&=&\frac{1}{N} \sum^{N-1}_{i=1}u^i X (u^\dagger)^i \,t, \nonumber \\
 B&=&\frac{1}{N} u^{N} X (u^\dagger)^{N} \,t, \nonumber
\end{eqnarray}
which results:
\begin{eqnarray}
 ||S_{N-1}-\ee^{\sum^N_{i=1}u^i X (u^\dagger)^i t/N}||_{\text{op}} \leqslant \sum^\infty_{i=2} \frac{2^{i+1}|t|^i}{i!N^i} ||Y||^i_{\text{op}} \sum^{i-1}_{j=1} 
 \left( {i \choose j}-1 \right). \label{2use}
\end{eqnarray}
Since $X \in \{Y-uYu^\dagger: Y \in {\cal B}(\mathcal{H})\}$, we obtain $P(X)=0$ and 
\begin{eqnarray}
 &&||\ee^{\sum^N_{i=1}u^i X (u^\dagger)^i t/N}-\ee^{P(X)t}||_{\text{op}} \leqslant \sum^\infty_{n=1} \frac{||\sum^N_{i=1}u^i X (u^\dagger)^i t/N||^n_{\text{op}}}{n!}= \nonumber \\
 &&=\sum^\infty_{n=1} \frac{||(u Y u^\dagger-u^{N+1} Y (u^\dagger)^{N+1}) t/N||^n_{\text{op}}}{n!} \leqslant 
 \sum^\infty_{n=1} \frac{\left(2 ||Y||_{\text{op}}\, |t|\right)^n}{n!N^n}. \label{3use}
\end{eqnarray}
Now, we are able to derive an upper bound for \eqref{start} by adding the three inequalities in \eqref{1use}, \eqref{2use} and \eqref{3use} 
and introducing the finite constant 
\begin{equation}
 M= 4 |t|^2 \ee^{2|t|\,||Y||_{\text{op}}}\,||Y||^2_{\text{op}}+ 2 ||Y||_{\text{op}}\,|t|, \nonumber
\end{equation}
which yields that for any $X \in \{Y-uYu^\dagger: Y \in {\cal B}(\mathcal{H})\}$ 
\begin{equation}
||\ee^{uXu^\dagger t/N} \ee^{u^2 X (u^\dagger)^2 t/N} \dots  \ee^{u^N X (u^\dagger)^N t/N}-\ee^{P(X)t}||_{\text{op}} 
\leqslant \frac{M}{N}+\mathcal{O}\left(\frac{1}{N^2}\right). \label{4use}
\end{equation}
The right-hand side of this inequality goes to zero as $N \to \infty$. 
Now, suppose $X \in F_{{\cal B}(\mathcal{H})}$ then the convergence of \eqref{limit1} is trivial:
\begin{eqnarray}
 ||\ee^{uXu^\dagger t/N} \ee^{u^2 X (u^\dagger)^2 t/N} \dots  \ee^{u^N X (u^\dagger)^N t/N}-\ee^{P(X)t}||_{\text{op}}=
 ||\ee^{X\,t} -\ee^{X\,t}||_{\text{op}}=0. \nonumber
\end{eqnarray}

The last part of the proof runs as follows. Theorem \ref{theorem1} states that any $X \in \Xi_{{\cal B}(\mathcal{H})}$ can be 
written as
\begin{equation}
 X=X_0 + W \nonumber
\end{equation}
where $X_0 \in F_{{\cal B}(\mathcal{H})}$ and $W \in \overline{\{X-uXu^\dagger: X \in {\cal B}(\mathcal{H})\}}$. 
Let $A,B,C \in {\cal B}(\mathcal{H})$ and
\begin{eqnarray}
&&||\ee^{B} \ee^{A+C}-\ee^{A} \ee^{B+C} ||_{\text{op}} \leqslant \sum^\infty_{i=2} \frac{2}{i!} \sum^{i-1}_{j=1} \left\{ 
 \left[ {i \choose j}-1 \right] ||A||^j_{\text{op}} ||C||^{i-j}_{\text{op}}+  \right. \nonumber \\
 && \left. + \left[ {i \choose j}-1 \right] ||B||^j_{\text{op}} ||C||^{i-j}_{\text{op}}+
 {i \choose j} ||A||^j_{\text{op}} ||B||^{i-j}_{\text{op}} \right\} \label{enuse},
\end{eqnarray}
which is derived by using the Taylor series of the exponentials. First, we consider an element 
$X \in \Xi_{{\cal B}(\mathcal{H})}$ such that $X=X_0+Y-uYu^\dagger$ and define the following operator:
\begin{equation}
 S_n=\ee^{n X_0 t/N}  \ee^{\sum^n_{i=1} u^i (Y-uYu^\dagger) (u^\dagger)^i t/N} \prod^N_{j=n+1} \ee^{u^j X (u^\dagger)^j t/N},\, 1\leqslant n < N. 
 \nonumber
\end{equation}
Then
\begin{eqnarray}
&&||\ee^{uXu^\dagger t/N} \ee^{u^2 X (u^\dagger)^2 t/N} \dots  \ee^{u^N X (u^\dagger)^N t/N}-\ee^{P(X)t}||_{\text{op}}\leqslant \label{bvc} \\
&&\leqslant ||\ee^{uXu^\dagger t/N} \ee^{u^2 X (u^\dagger)^2 t/N} \dots  \ee^{u^N X (u^\dagger)^N t/N}-S_1||_{\text{op}}+
 \sum^{N-2}_{n=1} ||S_n-S_{n+1}||_{\text{op}} \nonumber \\
&& + ||S_{N-1}-\ee^{X_0 t}  \ee^{\sum^N_{i=1} u^i (Y-uYu^\dagger) (u^\dagger)^i t/N}||_{\text{op}} +
||\ee^{X_0 t}  \ee^{\sum^N_{i=1} u^i (Y-uYu^\dagger) (u^\dagger)^i t/N}-\ee^{P(X)t}||_{\text{op}}. \nonumber
\end{eqnarray}
First, we have 
\begin{eqnarray}
 &&||\ee^{uXu^\dagger t/N} \ee^{u^2 X (u^\dagger)^2 t/N} \dots  \ee^{u^N X (u^\dagger)^N t/N}-S_1||_{\text{op}} \leqslant 
 \nonumber \\
 && \leqslant ||\ee^{uXu^\dagger t/N}-\ee^{X_0 t/N} \ee^{u (Y-uYu^\dagger) u^\dagger t/N}||_{\text{op}} \ee^{||X||_{\text{op}}\,|t|}
 \leqslant \nonumber \\
 && \leqslant \frac{2 |t|^2}{N^2}||X_0||_{\text{op}}\,||Y||_{\text{op}} \ee^{||X||_{\text{op}}\,|t|}+
\mathcal{O}\left(\frac{1}{N^3}\right) , \label{k1}
\end{eqnarray}
where we used inequality \eqref{ineq2} with $A=X_0 t/N$ and $B=u (Y-uYu^\dagger) u^\dagger t/ N$. In the next step we choose
\begin{eqnarray}
 A&=&\frac{X_0 t}{N}, \, \, B=\frac{1}{N} \sum^n_{i=1} u^i (Y-uYu^\dagger) (u^\dagger)^i t, \nonumber \\
 C&=& \frac{u^{n+1} (Y-uYu^\dagger) (u^\dagger)^{n+1} t}{N}, \nonumber 
\end{eqnarray}
where $1 \leqslant n < N-1$ and apply the result in \eqref{enuse}. Due to \eqref{ineqYeq} we have 
\begin{equation}
 ||A||_{\text{op}} \leqslant \frac{||X_0||_{\text{op}}\,|t|}{N},\quad ||B||_{\text{op}}, ||C||_{\text{op}} \leqslant
 \frac{2 ||Y||_{\text{op}}\,|t|}{N}. \nonumber
\end{equation}
Hence
\begin{eqnarray}
 &&||S_n-S_{n+1}||_{\text{op}}\leqslant \ee^{\frac{N-n-1}{N}||X||_{\text{op}}\,|t|} \times \nonumber \\
 \times &&||\ee^{n X_0 t/N}  \ee^{\sum^n_{i=1} u^i (Y-uYu^\dagger) (u^\dagger)^i t/N}\ee^{ u^{n+1} X (u^\dagger)^{n+1} t/N} -
 \ee^{(n+1) X_0 t/N}  \ee^{\sum^{n+1}_{i=1} u^i (Y-uYu^\dagger) (u^\dagger)^i t/N}||_{\text{op}}  \nonumber \\
\leqslant && ||\ee^{\sum^n_{i=1} u^i (Y-uYu^\dagger) (u^\dagger)^i t/N}\ee^{ u^{n+1} (X_0+Y-uYu^\dagger) (u^\dagger)^{n+1} t/N} -
 \ee^{X_0 t/N}  \ee^{\sum^{n+1}_{i=1} u^i (Y-uYu^\dagger) (u^\dagger)^i t/N}||_{\text{op}} \times \nonumber \\
\times &&\ee^{||X||_{\text{op}}\,|t|} \leqslant \ee^{||X||_{\text{op}}\,|t|} \frac{2|t|^2 }{N^2} ||Y||_{\text{op}}
\Big(2||Y||_{\text{op}}+ 3||X_0||_{\text{op}} \Big) + \mathcal{O}\left(\frac{1}{N^3}\right). \label{k2}
\end{eqnarray}
Similarly we obtain
\begin{equation}
 ||S_{N-1}-\ee^{X_0 t}  \ee^{\sum^N_{i=1} u^i (Y-uYu^\dagger) (u^\dagger)^i t/N}||_{\text{op}} \leqslant 
   \ee^{||X||_{\text{op}}\,|t|} \frac{2|t|^2 }{N^2} ||Y||_{\text{op}}
\Big(2||Y||_{\text{op}}+ 3||X_0||_{\text{op}} \Big) + \mathcal{O}\left(\frac{1}{N^3}\right)
 \label{k3}
\end{equation}
by choosing 
\begin{eqnarray}
 A&=&\frac{X_0 t}{N}, \, \, B=\frac{1}{N} \sum^{N-1}_{i=1} u^i (Y-uYu^\dagger) (u^\dagger)^i t, \nonumber \\
 C&=& \frac{u^{N} (Y-uYu^\dagger) (u^\dagger)^{N} t}{N}, \nonumber 
\end{eqnarray}
in \eqref{enuse}. For the last term  we get
\begin{eqnarray}
 &&||\ee^{X_0 t}  \ee^{\sum^N_{i=1} u^i (Y-uYu^\dagger) (u^\dagger)^i t/N}-\ee^{P(X)t}||_{\text{op}}=
 ||\ee^{X_0 t}  \ee^{\sum^N_{i=1} u^i (Y-uYu^\dagger) (u^\dagger)^i t/N}-\ee^{X_0 t}\ee^{P(Y-uYu^\dagger)t} ||_{\text{op}} \nonumber \\
 &&\leqslant \ee^{||X||_{\text{op}}\,|t|} || \ee^{\sum^N_{i=1} u^i (Y-uYu^\dagger) (u^\dagger)^i t/N}-\ee^{P(Y-uYu^\dagger)t} ||_{\text{op}}
 \leqslant \ee^{||X||_{\text{op}}\,|t|}  \frac{M}{N}+ \mathcal{O}\left(\frac{1}{N^2}\right)
 \label{k4}
\end{eqnarray}
where we used the result from \eqref{4use} and the fact that $P(Y-uYu^\dagger)=0$.

Substituting inequalities \eqref{k1}, \eqref{k2}, \eqref{k3}, and \eqref{k4} into \eqref{bvc}
results that for any element 
$X\in F_{{\cal B}(\mathcal{H})} \oplus \{Y-uYu^\dagger: Y \in {\cal B}(\mathcal{H})\}$
\begin{equation}
 ||\ee^{uXu^\dagger t/N} \ee^{u^2 X (u^\dagger)^2 t/N} \dots  \ee^{u^N X (u^\dagger)^N t/N}-\ee^{P(X)t}||_{\text{op}} \leqslant \frac{M'}{N} 
 + \mathcal{O}\left(\frac{1}{N^2}\right) \nonumber
\end{equation}
where we introduced the constant
\begin{equation}
 M'=\ee^{||X||_{\text{op}}\,|t|} \Big[M + 2|t|^2\, ||Y||_{\text{op}}
\Big(2||Y||_{\text{op}}+ 3||X_0||_{\text{op}} \Big)  \Big ] \nonumber
\end{equation}
and thus for any $X\in F_{{\cal B}(\mathcal{H})} \oplus \{Y-uYu^\dagger: Y \in {\cal B}(\mathcal{H})\}$
\begin{equation}
 ||\ee^{uXu^\dagger t/N} \ee^{u^2 X (u^\dagger)^2 t/N} \dots  \ee^{u^N X (u^\dagger)^N t/N}-\ee^{P(X)t}||_{\text{op}} \xrightarrow[N \to \infty]{} 
 0. \label{k5}
\end{equation}
We claim finally that for any $X\in F_{{\cal B}(\mathcal{H})} \oplus \overline{\{Y-uYu^\dagger: Y \in {\cal B}(\mathcal{H})\}}$ the theorem holds. 
Let us consider the following inequality for any $X_1,X_2 \in {\cal B}(\mathcal{H})$:
\begin{eqnarray}
&&||\ee^{X_1 t/N}-\ee^{X_2 t/N} ||_{\text{op}} \leqslant ||X_1-X_2 ||_{\text{op}}\,\frac{|t|}{N}+ 
||X^2_1-X^2_2 ||_{\text{op}} \frac{|t|^2}{2!N^2}+\dots \nonumber \\
&&\leqslant ||X_1-X_2 ||_{\text{op}} \left(\frac{|t|}{N}+\frac{||X_1||_{\text{op}}+||X_2||_{\text{op}}}{2!N^2}|t|^2+\dots \right)= 
\label{specuse2} \\
&&=||X_1-X_2 ||_{\text{op}} \frac{|t|}{N} \sum^\infty_{i=0} \frac{\left(||X_1||_{\text{op}}+||X_2||_{\text{op}}\right)^i}{(i+1)! N^i} |t|^i
\nonumber
\end{eqnarray}
where we used
\begin{eqnarray}
 &&||A^n-B^n ||_{\text{op}}=||A^n-A^{n-1}B+A^{n-1}B-A^{n-2}B^2+ \dots AB^{n-1}-B^n ||_{\text{op}} \nonumber \\
 &&\leqslant ||A-B ||_{\text{op}} \left( \sum^{n-1}_{i=0} ||A||^{n-1-i}_{\text{op}}||B||^{i}_{\text{op}}\right)\leqslant ||A-B ||_{\text{op}} 
 \Big(||A||_{\text{op}}+||B||_{\text{op}}\Big)^{n-1}. \nonumber
\end{eqnarray}
One more inequality is still required for this part of the proof. Let us introduce
\begin{eqnarray}
 P_0&=&\prod^{N}_{j=1} \ee^{u^j X_1 (u^\dagger)^j t/N} \nonumber \\
 P_i&=&\prod^{N-i}_{j=1} \ee^{u^j X_1 (u^\dagger)^j t/N}\prod^{N}_{j=N-i+1} \ee^{u^j X_2 (u^\dagger)^j t/N},\quad 0 < i <N. \nonumber
\end{eqnarray}
Then
\begin{eqnarray}
 &&||\ee^{uX_1u^\dagger t/N} \dots  \ee^{u^N X_1 (u^\dagger)^N t/N}-\ee^{u X_2 u^\dagger t/N} \dots  \ee^{u^N X_2 (u^\dagger)^N t/N}||_{\text{op}}
\nonumber \\
&& \leqslant \sum^{N-2}_{i=0}||P_i-P_{i+1}||_{\text{op}}+||P_{N-1}-\ee^{u X_2 u^\dagger t/N} \dots  \ee^{u^N X_2 (u^\dagger)^N t/N}||_{\text{op}}. \nonumber
\end{eqnarray}
First, we have
\begin{eqnarray}
||P_0-P_1||_{\text{op}} &\leqslant& \prod^{N-1}_{j=1} ||\ee^{u^j X_1 (u^\dagger)^j t/N}||_{\text{op}} || \ee^{u^N X_1 (u^\dagger)^N t/N}-
\ee^{u^N X_2 (u^\dagger)^N t/N}||_{\text{op}}
\nonumber \\
&\leqslant& \ee^{\frac{N-1}{N} ||X_1||_{\text{op}}\,|t|} || \ee^{X_1 t/N}-\ee^{X_2 t/N}||_{\text{op}}
\leqslant \ee^{||X_1||_{\text{op}}\,|t|} || \ee^{X_1 t/N}-\ee^{X_2 t/N}||_{\text{op}} \nonumber \\
&\leqslant& \ee^{ \left(||X_1||_{\text{op}}+ ||X_1-X_2 ||_{\text{op}} \right)\,|t|} \,
|| \ee^{X_1 t/N}-\ee^{X_2 t/N}||_{\text{op}}.\nonumber
\end{eqnarray}
We consider the following inequality
\begin{equation}
||X_2||_{\text{op}} \leqslant ||X_1||_{\text{op}} + ||X_1-X_2 ||_{\text{op}} \nonumber
\end{equation}
and then  for $N-1>i>0$
\begin{eqnarray}
&&||P_i-P_{i+1}||_{\text{op}} \leqslant \prod^{N-1-i}_{j=1} ||\ee^{u^j X_1 (u^\dagger)^j t/N}||_{\text{op}} 
|| \ee^{u^{N-i} X_1 (u^\dagger)^{N-i} t/N}-\ee^{u^{N-i} X_2 (u^\dagger)^{N-i} t/N}||_{\text{op}} \times
\nonumber \\
&&\times \prod^{N}_{j=N-i+1} ||\ee^{u^j X_2 (u^\dagger)^j t/N}||_{\text{op}} \leqslant 
\ee^{\frac{N-i-1}{N} ||X_1||_{\text{op}}\,|t|} || \ee^{X_1 t/N}-\ee^{X_2 t/N}||_{\text{op}}\, 
\ee^{\frac{i}{N} ||X_2||_{\text{op}}\,|t|} \nonumber \\
&& \leqslant \ee^{ \left(||X_1||_{\text{op}}+ ||X_1-X_2||_{\text{op}} \right)\,|t|} \,
|| \ee^{X_1 t/N}-\ee^{X_2 t/N}||_{\text{op}}.
\nonumber
\end{eqnarray}
We also have
\begin{eqnarray}
&&||P_{N-1}-\ee^{uX_2u^\dagger t/N} \dots  \ee^{u^N X_2 (u^\dagger)^N t/N}||_{\text{op}}\leqslant  
 || \ee^{X_1 t/N}-\ee^{X_2 t/N}||_{\text{op}}\, 
\ee^{\frac{N-1}{N} ||X_2||_{\text{op}}\,|t|} \nonumber \\
&& \leqslant \ee^{ \left(||X_1||_{\text{op}}+ ||X_1-X_2 ||_{\text{op}} \right)\,|t|} \,
|| \ee^{X_1 t/N}-\ee^{X_2 t/N}||_{\text{op}}.
\nonumber
\end{eqnarray}
Thus for all $N > 1$
\begin{eqnarray}
 &&||\ee^{uX_1u^\dagger t/N} \dots  \ee^{u^N X_1 (u^\dagger)^N t/N}-\ee^{uX_2u^\dagger t/N} \dots  \ee^{u^N X_2 (u^\dagger)^N t/N}||_{\text{op}}
\label{specuse} \\
&& \leqslant N \ee^{ \left(||X_1||_{\text{op}}+ ||X_1-X_2 ||_{\text{op}} \right)\,|t|} \,
|| \ee^{X_1 t/N}-\ee^{X_2 t/N}||_{\text{op}}. \nonumber
\end{eqnarray}
Let $X_1=X_0+W$ and $W$ be not of the form $X-uXu^\dagger$ and $ t \in \mathbb{C}$ with $|t|<\infty$. Then for any $\epsilon > 0$ we can take 
an arbitrary $\epsilon' > 0$ satisfying 
\begin{equation}
\ee^{ \left(||X_1||_{\text{op}}+ \epsilon' \right)\,|t|} |t| \epsilon' (1+ \epsilon')+ \epsilon' \leqslant \epsilon, \nonumber
\end{equation}
and for this $\epsilon'$ there exists $X_2 \in {\cal B}(\mathcal{H})$ and $Y \in {\cal B}(\mathcal{H})$ such that
$X_2=X_0 + Y-uYu^\dagger$ and
\begin{equation}
 ||X_1-X_2||_{\text{op}}< \epsilon'. \nonumber
\end{equation} 
Due to the convergence in \eqref{k5} we have that for $\epsilon'> 0$ we can take an $N'$ such that for every $N>N'$ 
\begin{equation}
||\ee^{uX_2u^\dagger t/N} \ee^{u^2 X_2 (u^\dagger)^2 t/N} \dots  \ee^{u^N X_2 (u^\dagger)^N t/N}-
\ee^{P(X_1)t}||_{\text{op}}< \epsilon' \label{nm}
\end{equation}
where we used the fact that $P(X_1)=P(X_2)=X_0$. There exists an $N''$ such that for every $N>N''$
\begin{equation}
\sum^\infty_{i=1} \frac{\left(||X_1||_{\text{op}}+||X_2||_{\text{op}} \right)^i}{(i+1)! N^i}|t|^i < \epsilon' \nonumber
\end{equation}
and together with \eqref{specuse2} we get
\begin{equation}
 N \ee^{ \left(||X_1||_{\text{op}}+ ||X_1-X_2 ||_{\text{op}} \right)\,|t|} \,
|| \ee^{X_1 t/N}-\ee^{X_2 t/N}||_{\text{op}} < \ee^{\left(||X_1||_{\text{op}}+ \epsilon' \right)\,|t|} |t| \epsilon' (1+ \epsilon').
\nonumber
\end{equation}
Thus, with the aid of \eqref{specuse} and \eqref{nm} we have that for all $N>\max\{N',N''\}$ 
\begin{eqnarray}
 &&||\ee^{uX_1u^\dagger t/N} \ee^{u^2 X_1 (u^\dagger)^2 t/N} \dots  \ee^{u^N X_1 (u^\dagger)^N t/N}-\ee^{P(X_1)t}||_{\text{op}}  \nonumber \\
 && \leqslant ||\ee^{uX_1u^\dagger t/N} \dots  \ee^{u^N X_1 (u^\dagger)^N t/N}-\ee^{uX_2u^\dagger t/N} \dots  \ee^{u^N X_2 (u^\dagger)^N t/N}||_{\text{op}}+ \nonumber \\
 && + ||\ee^{uX_2u^\dagger t/N} \ee^{u^2 X_2 (u^\dagger)^2 t/N} \dots  \ee^{u^N X_2 (u^\dagger)^N t/N}-
 \ee^{P(X_1)t}||_{\text{op}} \nonumber \\
  && <  \ee^{\left(||X_1||_{\text{op}}+ \epsilon' \right)\,|t|} |t| \epsilon' (1+ \epsilon')+\epsilon' \leqslant \epsilon. \nonumber
\end{eqnarray}
Therefore 
\begin{equation}
  \lim_{N \to \infty}||\ee^{uXu^\dagger t/N} \ee^{u^2 X (u^\dagger)^2 t/N} \dots  \ee^{u^N X (u^\dagger)^N t/N}-\ee^{P(X)t}||_{\text{op}} = 0 \nonumber
 \end{equation}
holds for all $X \in \Xi_{{\cal B}(\mathcal{H})}$, which yields the desired conclusion.
\end{proof}     

The above results show for the series in Eq. \eqref{Ueq} that
\begin{equation}
 \lim_{N \to \infty} ||u_N(t)-\ee^{-iH_{\text{id}}t}u^N||_{\text{op}}=0, \nonumber
\end{equation}
which holds for the Hamilton operator $H=H_{\text{id}}+H_{\text{err}} \in \Xi_{{\cal B}(\mathcal{H})}$ with
\begin{equation}
 [H_{\text{id}},u]=0,\,\,\, [H_{\text{err}},u]\neq0. \nonumber
\end{equation}
Theorem \ref{theorem2} also shows that for pulse numbers $N \gg ||H||_{\text{op}}t$ the unitary operator $u_N(t)$
approximates well the ideal evolution $\ee^{-iH_{\text{id}}t}$. This resembles the experimental findings, whereas
longer coherence times require more pulses.

\section{Non-equidistant dynamical control}
\label{III}

In the case of the equidistant dynamical control the time interval $[0,t]$ is split into $N$ equal time intervals with length
$t/N$. Now, the time interval is split again into $N$ time intervals with lengths $t_1,t_2,...,t_N$ where $t_i \geqslant 0$ for all $i$
and $\sum^N_{i=1}t_i=t$. In order to study the limit $N \to \infty$ we introduce the matrix $(a)_{i,j}$ containing the weights of the splitting for all $N$, 
such that $a_{N,i} \in [0,1)$ when $i\leqslant N$ otherwise $a_{N,i}=0$ and $\sum^N_{i=1}a_{N,i}=1$. The time evolution after applying $N$ pulses 
is governed by the unitary operator
\begin{equation}
 u^{(a)}_N(t)=u \ee^{-ia_{N,1}Ht} u \ee^{-ia_{N,2}Ht} \dots u \ee^{-ia_{N,N}Ht}.
 \label{Uneq}
\end{equation}
We express every $\ee^{-iHa_{N,i}t}$ term by its infinite series and obtain
\begin{eqnarray}
 u^{(a)}_N(t)= u^N-it\left(\sum^N_{i=1} a_{N,i} u^i H (u^\dagger)^i\right)u^N + \dots \nonumber
\end{eqnarray}
Therefore, a natural problem is that of determining the convergence of the weighted Ces\`{a}ro mean:
\begin{equation}
\sum^N_{i=1} a_{N,i} u^i H (u^\dagger)^i=\sum^N_{i=1} a_{N,i} T^i(H),
\label{noneqmean}
\end{equation}
where $T$ is a linear isometry on the Banach space ${\cal B}(\mathcal{H})$, as it has been shown in Proposition \ref{prop1}. 
Most of the next mean ergodic theorem is due to L. W. Cohen\cite{Cohen}:
\begin{theorem}[Cohen's ergodic theorem]
 \label{theorem3}
 If $T$ is a power bounded linear operator on a Banach space ${\cal B}$, $(a)_{i,j}$ is a matrix such that 
 $a_{N,i} \geqslant 0$, $\sum^\infty_{i=1}a_{N,i}=1$ for all $N$, $\lim_{N \to \infty}a_{N,i}=0$ for all $i$,
 \begin{equation}
  \lim_{k \to \infty} \sum^\infty_{i=k} |a_{N,i+1}-a_{N,i}|=0 
  \label{critcond}
 \end{equation}
uniformly in $N$ then 
\begin{equation}
  \Xi_{{\cal B}}=F_{{\cal B}} \oplus \overline{\{x-Tx: x \in {\cal B}\}}. \nonumber
 \end{equation}
The linear operator $Px=\lim \sum^\infty_{i=1} a_{N,i}T^ix$ for $x \in \Xi_{{\cal B}}$ is the projection of $\Xi_{{\cal B}}$ onto $F_{{\cal B}}$.
We have $P=P^2=TP=PT$ and for any $z\in {\cal B}$ the assertions
\begin{eqnarray}
 \text{1)}\,\,\, \lim \sum^\infty_{i=1} a_{N,i}T^iz=0, \nonumber \\
 \text{2)}\,\,\,  z \in \overline{\{x-Tx: x \in {\cal B}\}} \nonumber 
\end{eqnarray}
are equivalent.
\end{theorem}

For the linear isometry $T(X)=u X u^\dagger$ on ${\cal B}(\mathcal{H})$ we have
\begin{equation}
 ||\sum^\infty_{i=1} a_{N,i}T^i(X)||_{\text{op}}\leqslant \sum^\infty_{i=1} a_{N,i} ||X||_{\text{op}}=||X||_{\text{op}}. \nonumber
\end{equation}

The weights $a_{i,j}$ of the time's splitting in a dynamical control scheme fulfil the condition $\lim_{N \to \infty} a_{N,i}=0$, 
because the increase of the number of pulses $N$ implies shorter time intervals between the applications of the pulses.
A key element of the proof is the following relation
\begin{equation}
 \lim_{N \to \infty} \left(a_{N,1}+\sum^{N-1}_{i=1} |a_{N,i+1}-a_{N,i}|+a_{N,N} \right) =0, 
 \label{keyCohen}
\end{equation}
which is a consequence of the conditions in Theorem \ref{theorem3}.
\begin{example} 
Let us define matrix $(a)_{i,j}$ in the following way:
\begin{eqnarray}
 a_{N,1}=\frac{2N-1}{N^2}, a_{N,2}=\frac{1}{N^2}, \dots, a_{N,N-1}=\frac{2N-1}{N^2}, a_{N,N}=\frac{1}{N^2}, \,\, N \,\text{is even}, \nonumber
\end{eqnarray}
and 
\begin{eqnarray}
 a_{N,1}=\frac{2N-1}{N^2}, a_{N,2}=\frac{1}{N^2}, \dots, a_{N,N-1}=\frac{1}{N^2}, a_{N,N}=\frac{1}{N}, \,\, N \, \text{is odd}. \nonumber
\end{eqnarray}
This matrix fulfils the following conditions of Theorem \ref{theorem3}:
\begin{eqnarray}
a_{N,i} \geqslant 0,\, \sum^\infty_{i=1}a_{N,i}=1, \, \forall N, \,\, \lim_{N \to \infty}a_{N,i}=0, \forall i \nonumber
\end{eqnarray}
but 
\begin{equation}
 \lim_{k \to \infty} \sum^\infty_{i=k} |a_{N,i+1}-a_{N,i}|=0 \nonumber
\end{equation}
is not uniformly in $N$. Eq. \eqref{keyCohen} yields
\begin{equation}
 \lim_{N \to \infty} \left(a_{N,1}+\sum^{N-1}_{i=1} |a_{N,i+1}-a_{N,i}|+a_{N,N} \right) =2. \nonumber
\end{equation}
Thus if we negate proposition \eqref{critcond} but we keep the conditions
\begin{eqnarray}
 &&a_{N,i} \geqslant 0, \quad \sum^\infty_{i=1}a_{N,i}=1, \quad \forall N, \nonumber \\
 &&\lim_{N \to \infty}a_{N,i}=0, \quad \forall i, \nonumber
\end{eqnarray}
then this implies that proposition \eqref{keyCohen} fails. This shows that 
\begin{equation}
  \lim_{k \to \infty} \sum^\infty_{i=k} |a_{N,i+1}-a_{N,i}|=0 \nonumber
 \end{equation}
uniformly in $N$ is not only sufficient but also necessary for the proof of Theorem  \ref{theorem3}.
\end{example}

\begin{example}
 Let $f:[0,1] \rightarrow \mathbb R^+$ be Riemann integrable with $\int^1_0 f(x)dx=1$ and 
\begin{equation}
 a_{N,i}= \int^{\frac{i}{N}}_{\frac{i-1}{N}} f(x)dx. \nonumber
\end{equation}
First we have
\begin{equation}
 a_{N,1}+a_{N,N}\leqslant \left (\sup_{x \in [0, \frac{1}{N}]} f(x)+\sup_{x \in [\frac{N-1}{N}, 1]} f(x) \right)/N \nonumber
\end{equation}
which tends to zero as $N \to \infty$, because $f$ is bounded on the interval $[0,1]$. Now for any $1 \leqslant i < N$ 
\begin{equation}
 |a_{N,i+1}-a_{N,i}|\leqslant 2\left (\sup_{x \in [\frac{i-1}{N}, \frac{i+1}{N}]} f(x)-\inf_{x \in [\frac{i-1}{N}, \frac{i+1}{N}]} f(x) \right)/N. \nonumber
\end{equation}
Due to the fact that the upper and lower Riemann sums of $f$ with respect to the partition $\{[0, \frac{2}{N}], \dots ,[\frac{N-2}{N},1]\}$ of $[0,1]$ 
tend to the same value as $N \to \infty$, we have
\begin{equation}
 \lim_{N \to \infty} \sum^{N-1}_{i=1}|a_{N,i+1}-a_{N,i}|=0. \nonumber
\end{equation}
In the case of Uhrig's dynamical decoupling $f(x)=\frac{\pi}{2} \sin(\pi x)$ \cite{Uhrig}.
\end{example}

The following theorem extends the results of Theorem \ref{theorem2}.
\begin{theorem}
 \label{theorem4}
 Let $u$ be a unitary operator in the Hilbert space $\mathcal{H}$ and the set $\Xi_{{\cal B}(\mathcal{H})}$ is defined
 by the power bounded linear map $T(X)=u X u^\dagger$. Let $P$ be the projection 
 operator which maps $\Xi_{{\cal B}(\mathcal{H})}$ onto the linear subspace $F_{{\cal B}(\mathcal{H})}=\{X \in {\cal B}(\mathcal{H}): [u,X]=0\}$. 
  Suppose a matrix $(a)_{i,j}$ such that 
 $1>a_{N,i} \geqslant 0$, $\sum^N_{i=1}a_{N,i}=1$ for all $N$, $\lim_{N \to \infty}a_{N,i}=0$ for all $i$ and
 \begin{equation}
  \lim_{k \to \infty} \sum^\infty_{i=k} |a_{N,i+1}-a_{N,i}|=0
  \label{critcond2}
 \end{equation}
uniformly in $N$. Then, for any $X \in \Xi_{{\cal B}(\mathcal{H})}$ 
 and $t \in \mathbb{C}$ with $|t|<\infty$
  \begin{equation}
  \lim_{N \to \infty} ||u \ee^{a_{N,1}Xt} u \ee^{a_{N,2}Xt} \dots u \ee^{a_{N,N}Xt}-\ee^{P(X)t} u^N||_{\text{op}} = 0.
\label{limit2}
 \end{equation}
\end{theorem}
\begin{proof}
Here, we are going to mimic the proof of Theorem \ref{theorem2} to obtain the statement of \eqref{limit2}. Therefore, we start with
$X \in \{Y-uYu^\dagger: Y \in {\cal B}(\mathcal{H})\}$  which means that there exists a $Y \in {\cal B}(\mathcal{H})$ such that $X=Y-uYu^\dagger$ and
\begin{eqnarray}
\label{ineqYeqn}
||\sum^n_{i=1}a_{N,i}u^i (Y-uYu^\dagger) (u^\dagger)^i ||_{\text{op}}\leqslant ||Y||_{\text{op}} \left(a_{N,1}+\sum^{n-1}_{i=1}|a_{N,i+1}-a_{N,i}|+a_{N,n} \right), \nonumber \\ 
\end{eqnarray}
for $n>1$ where used the unitarily invariant property of $||.||_{\text{op}}$.

Let
\begin{equation}
 S_n=\ee^{\sum^n_{i=1}a_{N,i}u^i X (u^\dagger)^i t} \prod^N_{j=n+1} \ee^{a_{N,j}u^j X (u^\dagger)^j t},\, 1\leqslant n < N \nonumber
\end{equation}
and by applying inequality \eqref{ineq2} for $A= a_{N,1}u X u^\dagger t$ and $B=a_{N,2}u^2 X (u^\dagger)^2 t$ we obtain 
\begin{eqnarray}
||S_1-S_2||_{\text{op}}\leqslant \ee^{|t| \sum^N_{i=3} a_{N,i} ||X||_{\text{op}}} \sum^\infty_{i=2} \frac{2 |t|^i}{i!} C^{(i,1)} ||Y||^i_{\text{op}} 
\leqslant \ee^{2 |t|\, ||Y||_{\text{op}}} \sum^\infty_{i=2} \frac{2 |t|^i}{i!} C^{(i,1)} ||Y||^i_{\text{op}} \nonumber
\end{eqnarray}
with 
\begin{eqnarray}
C^{(i,1)}=2^i \sum^{i-1}_{j=1} 
 \left( {i \choose j}-1 \right) a^j_{N,1} a^{i-j}_{N,2}. \label{Cels}
\end{eqnarray}
For $1<n<N-1$, where $A=\sum^n_{i=1}a_{N,i}u^i X (u^\dagger)^i t$ and $B=a_{N,n+1}u^{n+1} X (u^\dagger)^{n+1} t$ are substituted in \eqref{ineq2} we have
\begin{eqnarray}
||S_n-S_{n+1}||_{\text{op}}\leqslant \ee^{|t| \sum^N_{i=n+2} a_{N,i} ||X||_{\text{op}}} \sum^\infty_{i=2} \frac{2 |t|^i}{i!} C^{(i,n)} ||Y||^i_{\text{op}} 
\leqslant \ee^{2|t|\,||Y||_{\text{op}}} \sum^\infty_{i=2} \frac{2 |t|^i}{i!} C^{(i,n)} ||Y||^i_{\text{op}} \nonumber 
\end{eqnarray}
with
\begin{eqnarray}
C^{(i,n)}= \sum^{i-1}_{j=1} 
 \left( {i \choose j}-1 \right)  2^{i-j} a^{i-j}_{N,n+1} \left(a_{N,1}+\sum^{n-1}_{k=1} |a_{N,k+1}-a_{N,k}|+a_{N,n} \right)^j. \nonumber \\ 
 \label{Calt}
\end{eqnarray}
Furthermore,
\begin{eqnarray}
 ||S_{N-1}-\ee^{\sum^N_{i=1}a_{N,i}u^i X (u^\dagger)^i t}||_{\text{op}} \leqslant \sum^\infty_{i=2} \frac{2 |t|^i}{i!} C^{(i,N-1)} ||Y||^i_{\text{op}}, \nonumber 
\end{eqnarray}
and
\begin{eqnarray}
 ||\ee^{\sum^N_{i=1}a_{N,i}u^i X (u^\dagger)^i t}-\ee^{P(X)t}||_{\text{op}} \leqslant 
 \sum^\infty_{n=1} \frac{\left(||Y||_{\text{op}}(a_{N,1}+\sum^{N-1}_{i=1} |a_{N,i+1}-a_{N,i}|+a_{N,N})|t|\right)^n}{n!}. \nonumber
\end{eqnarray}
We used the fact that $P(X)=P(Y-uYu^\dagger)=0$. Thus
\begin{eqnarray}
 &&||\ee^{a_{N,1}uXu^\dagger t} \ee^{a_{N,2}u^2 X (u^\dagger)^2 t} \dots  \ee^{a_{N,N}u^N X (u^\dagger)^N t}-\ee^{P(X)t}||_{\text{op}}\leqslant \nonumber \\
&& \leqslant \sum^{N-2}_{n=1} ||S_n-S_{n+1}||_{\text{op}}+||S_{N-1}-\ee^{\sum^N_{i=1}a_{N,i}u^i X (u^\dagger)^i t}||_{\text{op}} \nonumber \\
&& + ||\ee^{\sum^N_{i=1}a_{N,i}u^i X (u^\dagger)^i t}-\ee^{P(X)t}||_{\text{op}}\leqslant  \nonumber \\
&&  \leqslant \ee^{2 |t| \, ||Y||_{\text{op}}} \sum^{N-2}_{k=1} \sum^\infty_{i=2} \frac{2 |t|^i}{i!} C^{(i,k)} ||Y||^i_{\text{op}}+\sum^\infty_{i=2} 
\frac{2 |t|^i}{i!} C^{(i,N-1)} ||Y||^i_{\text{op}} \nonumber \\
 &&+\sum^\infty_{n=1} \frac{\left(||Y||_{\text{op}}(a_{N,1}+\sum^{N-1}_{i=1} |a_{N,i+1}-a_{N,i}|+a_{N,N})|t|\right)^n}{n!}.
 \label{optim}
\end{eqnarray}
We consider an $\epsilon > 0$ and by the hypothesis of uniformity in \eqref{critcond2} there is a $k_\epsilon$ such that
\begin{equation}
 a_{N,1}+\sum^{N-1}_{k=1} |a_{N,k+1}-a_{N,k}|+a_{N,N} \leqslant 2 \sum^{k_\epsilon}_{k=1} a_{N,k}+\epsilon \nonumber
\end{equation}
for $n>1$ and due to the condition $\lim_{N \to \infty}a_{N,i}=0$ there is an $N_\epsilon$ such that
\begin{equation}
 2\sum^{k_\epsilon}_{k=1} a_{N,k} < \epsilon,\,\,\, N>N_\epsilon. \nonumber
\end{equation}
Thus
\begin{equation}
a_{N,1}+\sum^{N-1}_{i=1} |a_{N,i+1}-a_{N,i}|+a_{N,N}<2 \epsilon, \,\,\, N>N_\epsilon \nonumber  
\end{equation}
which is equivalent to (see also \cite{Jardas})
\begin{equation}
 \lim_{N \to \infty} \sum^N_{i=1} |a_{N,i+1}-a_{N,i}|=0. \label{1use2}
\end{equation}
We have for $i>1$ and $N>n>1$
\begin{eqnarray}
C^{(i,n)}&=& \sum^{i-1}_{j=1} 
 \left( {i \choose j}-1 \right)  2^{i-j} a^{i-j}_{N,n+1} \left(a_{N,1}+\sum^{n-1}_{k=1} |a_{N,k+1}-a_{N,k}|+a_{N,n} \right)^j \nonumber \\
 &<& \Big( a_{N,1}+\sum^{n-1}_{k=1} |a_{N,k+1}-a_{N,k}|+a_{N,n} + 2 a_{N,n+1} \Big)^i \nonumber
\end{eqnarray}
and due to \eqref{1use2} and the properties of matrix $(a)_{i,j}$ we see that $C^{(i,n)} \xrightarrow[N \to \infty]{} 0 $. When $n=1$, then
\begin{equation}
 C^{(i,1)} < 2^i(a_{N,1}+a_{N,2})^i \xrightarrow[N \to \infty]{} 0. \nonumber
\end{equation}
Let us introduce 
\begin{equation}
 x_N =
      \begin{cases}
        \max\{a_{N,1}+\sum^{n-1}_{k=1} |a_{N,k+1}-a_{N,k}|+a_{N,n}, a_{N,n+1} \} & N>n>1, \\
        \max\{a_{N,1}, a_{N,2} \} & n=1,
      \end{cases}
      \nonumber
\end{equation}
which has the property $x_N \xrightarrow[N \to \infty]{} 0$. For $i>2$ we get
\begin{equation}
 \frac{C^{(i,1)}}{C^{(2,1)}} < 2^{i-2} x^{i-2}_N, \quad  \frac{C^{(i,n)}}{C^{(2,n)}} < 3^{i} x^{i-2}_N \nonumber
\end{equation}
which means that the terms $C^{(2,n)}$ are the slowest to converge to zero. Therefore we consider in Eq. \eqref{optim}
the case of  $i=2$, which yields
\begin{eqnarray}
 \sum^{N-2}_{k=1} C^{(2,k)}&=&4 a_{N,2}a_{N,1}+2 a_{N,3} \left(a_{N,1}+|a_{N,2}-a_{N,1}|+a_{N,2} \right)+\dots +\nonumber \\
 &+&2a_{N,N-1}\left(a_{N,1}+\sum^{N-3}_{i=1} |a_{N,i+1}-a_{N,i}|+a_{N,{N-2}}\right)\leqslant \nonumber \\
 &=&2 a_{N,2}a_{N,1}+2 \sum^{N-1}_{i=2}a_{N,i} a_{N,1}+2 \sum^{N-1}_{i=3}a_{N,i} |a_{N,2}-a_{N,1}|+ \dots \leqslant \nonumber \\
 &\leqslant&2 a_{N,2}a_{N,1}+2a_{N,1}+2\sum^{N-3}_{i=1} |a_{N,i+1}-a_{N,i}|+2a_{N,N-2} \xrightarrow[N \to \infty]{} 0, \nonumber
\end{eqnarray}
where we used the relation $\sum^{N-1}_{i=k} a_{N,i} \leqslant 1 $ for all $1\leqslant k \leqslant N-1$ and $\lim_{N \to \infty} a_{N,i}=0$ for all $i$. Now, combining all 
these results in \eqref{optim} we find that
for all $X \in \{Y-uYu^\dagger: Y \in {\cal B}(\mathcal{H})\}$
\begin{eqnarray}
 &&||\ee^{a_{N,1}uXu^\dagger t} \ee^{a_{N,2}u^2 X (u^\dagger)^2 t} \dots  \ee^{a_{N,N}u^N X (u^\dagger)^N t}-\ee^{P(X)t}||_{\text{op}} 
 \xrightarrow[N \to \infty]{} 0. \label{b1}
\end{eqnarray}
If $X \in F_{{\cal B}(\mathcal{H})}$ then the convergence of \eqref{limit2} is trivial:
\begin{equation}
||\ee^{a_{N,1}uXu^\dagger t} \ee^{a_{N,2}u^2 X (u^\dagger)^2 t} \dots  \ee^{a_{N,N}u^N X (u^\dagger)^N t}-\ee^{P(X)t}||_{\text{op}}=
||\ee^{\sum^N_{i=1}a_{N,i}Xt}-\ee^{X\,t}||_{\text{op}}=0. \nonumber
\end{equation}
In the last part of the proof we consider all $X \in \Xi_{{\cal B}(\mathcal{H})}$. According to Theorem \ref{theorem3} 
\begin{equation}
 X=X_0 + W \nonumber
\end{equation}
where $X_0 \in F_{{\cal B}(\mathcal{H})}$ and $W \in \overline{\{X-uXu^\dagger: X \in {\cal B}(\mathcal{H})\}}$. First, we take 
an element $X \in \Xi_{{\cal B}(\mathcal{H})}$ such that $X=X_0+Y-uYu^\dagger$. In order to use the result in \eqref{enuse} we
introduce
\begin{equation}
 S_n=\ee^{\sum^n_{i=1} a_{N,i} X_0 t}  \ee^{\sum^n_{i=1} a_{N,i} u^i (Y-uYu^\dagger) (u^\dagger)^i t} 
 \prod^N_{j=n+1}  \ee^{a_{N,j} u^j X (u^\dagger)^j t},\, 1\leqslant n < N. 
 \nonumber
\end{equation}
Then
\begin{eqnarray}
&&||\ee^{a_{N,1} uXu^\dagger t} \ee^{a_{N,2} u^2 X (u^\dagger)^2 t} \dots  \ee^{a_{N,N} u^N X (u^\dagger)^N t}-\ee^{P(X)t}||_{\text{op}}
\leqslant \label{bvc2} \\
&&\leqslant ||\ee^{a_{N,1} uXu^\dagger t} \ee^{a_{N,2} u^2 X (u^\dagger)^2 t} \dots  \ee^{a_{N,N} u^N X (u^\dagger)^N t}-S_1||_{\text{op}}+
 \sum^{N-2}_{n=1} ||S_n-S_{n+1}||_{\text{op}} \nonumber \\
&& + ||S_{N-1}-\ee^{X_0 t}  \ee^{\sum^N_{i=1} a_{N,i} u^i (Y-uYu^\dagger) (u^\dagger)^i t}||_{\text{op}} +
||\ee^{X_0 t}  \ee^{\sum^N_{i=1} a_{N,i} u^i (Y-uYu^\dagger) (u^\dagger)^i t}-\ee^{P(X)t}||_{\text{op}}. \nonumber
\end{eqnarray}
Substituting $A=a_{N,1} X_0 t$ and $B=a_{N,1} u (Y-uYu^\dagger) u^\dagger t$ in \eqref{ineq2} we get
\begin{eqnarray}
 &&||\ee^{uXu^\dagger t/N} \ee^{u^2 X (u^\dagger)^2 t/N} \dots  \ee^{u^N X (u^\dagger)^N t/N}-S_1||_{\text{op}} \leqslant 
 \nonumber \\
 && \leqslant ||\ee^{uXu^\dagger t/N}-\ee^{X_0 t/N} \ee^{u (Y-uYu^\dagger) u^\dagger t/N}||_{\text{op}} \ee^{||X||_{\text{op}}\,|t|}
 \leqslant \nonumber \\
 && \leqslant 2 a^2_{N,1} |t|^2||X_0||_{\text{op}}\,||Y||_{\text{op}} \ee^{||X||_{\text{op}}\,|t|}+
\mathcal{O}\left(a^3_{N,1}\right). \label{k12}
\end{eqnarray}
We set
\begin{eqnarray}
 A&=&a_{N,n+1} X_0 t, \, \, B=\sum^n_{i=1} a_{N,i} u^i (Y-uYu^\dagger) (u^\dagger)^i t, \nonumber \\
 C&=&  a_{N,n+1} u^{n+1} (Y-uYu^\dagger) (u^\dagger)^{n+1} t, \nonumber 
\end{eqnarray}
in \eqref{enuse} where $1 < n < N-1$ and we obtain
\begin{eqnarray}
 &&||S_n-S_{n+1}||_{\text{op}}\leqslant \ee^{||X||_{\text{op}}\,|t|} \times \nonumber \\
 \times &&||  \ee^{\sum^n_{i=1} a_{N,i} u^i (Y-uYu^\dagger) (u^\dagger)^i t}\ee^{ a_{N,n+1} u^{n+1} X (u^\dagger)^{n+1} t} -
 \ee^{a_{N,n+1} X_0 t}  \ee^{\sum^{n+1}_{i=1} a_{N,i} u^i (Y-uYu^\dagger) (u^\dagger)^i t}||_{\text{op}}  \nonumber \\
\leqslant && \ee^{||X||_{\text{op}}\,|t|}\sum^\infty_{i=2} \frac{2 |t|^i}{i!} \Big \{ C^{(i,n)} ||Y||^i_{\text{op}} + 
\sum^{i-1}_{j=1}  \left[ {i \choose j}-1 \right]  2^{i-j} a^i_{N,n+1} ||X_0||^j_{\text{op}} ||Y||^{i-j}_{\text{op}}
\nonumber \\
 + && \sum^{i-1}_{j=1} 
 {i \choose j}  2^{i-j} a^{j}_{N,n+1} \left(a_{N,1}+\sum^{n-1}_{k=1} |a_{N,k+1}-a_{N,k}|+a_{N,n} \right)^{i-j} 
 ||X_0||^j_{\text{op}} ||Y||^{i-j}_{\text{op}} \Big \},
\label{k22}
\end{eqnarray}
where we used the definition in \eqref{Calt}. In the case when $n=1$
\begin{eqnarray}
 &&||S_1-S_2||_{\text{op}}\leqslant  \nonumber \\
\leqslant && \ee^{||X||_{\text{op}}\,|t|}\sum^\infty_{i=2} \frac{2 |t|^i}{i!} \Big \{ C^{(i,1)} ||Y||^i_{\text{op}} + 
\sum^{i-1}_{j=1}  \left[ {i \choose j}-1 \right]  2^{i-j} a^i_{N,2} ||X_0||^j_{\text{op}} ||Y||^{i-j}_{\text{op}}
\nonumber \\
 + && \sum^{i-1}_{j=1} 
 {i \choose j}  2^{i-j} a^{j}_{N,2} a^{i-j}_{N,1} ||X_0||^j_{\text{op}} ||Y||^{i-j}_{\text{op}} \Big \}.
\label{k22g}
\end{eqnarray}
We also have
\begin{eqnarray}
 &&||S_{N-1}-\ee^{X_0 t}  \ee^{\sum^N_{i=1} a_{N,i} u^i (Y-uYu^\dagger) (u^\dagger)^i t}||_{\text{op}} \leqslant \nonumber \\
\leqslant && \ee^{||X||_{\text{op}}\,|t|}\sum^\infty_{i=2} \frac{2 |t|^i}{i!} \Big \{ C^{(i,N-1)} ||Y||^i_{\text{op}} + 
\sum^{i-1}_{j=1}  \left[ {i \choose j}-1 \right]  2^{i-j} a^i_{N,N} ||X_0||^j_{\text{op}} ||Y||^{i-j}_{\text{op}}
\nonumber \\
 + && \sum^{i-1}_{j=1} 
 {i \choose j}  2^{i-j} a^{j}_{N,N} \left(a_{N,1}+\sum^{N-2}_{k=1} |a_{N,k+1}-a_{N,k}|+a_{N,N-1} \right)^{i-j} 
 ||X_0||^j_{\text{op}} ||Y||^{i-j}_{\text{op}} \Big \} \nonumber \\ \label{k32}
\end{eqnarray}
by choosing 
\begin{eqnarray}
 A&=&a_{N,N} X_0 t, \, \, B=\sum^{N-1}_{i=1} a_{N,i} u^i (Y-uYu^\dagger) (u^\dagger)^i t, \nonumber \\
 C&=&  a_{N,N} u^{N} (Y-uYu^\dagger) (u^\dagger)^{N} t, \nonumber 
\end{eqnarray}
in \eqref{enuse}. For the last term we get
\begin{eqnarray}
 &&||\ee^{X_0 t}  \ee^{\sum^N_{i=1} a_{N,i} u^i (Y-uYu^\dagger) (u^\dagger)^i t}-\ee^{P(X)t}||_{\text{op}}=\nonumber \\
 &&=||\ee^{X_0 t}  \ee^{\sum^N_{i=1} a_{N,i} u^i (Y-uYu^\dagger) (u^\dagger)^i t}-\ee^{X_0 t}\ee^{P(Y-uYu^\dagger)t} ||_{\text{op}} \nonumber \\
 &&\leqslant \ee^{||X||_{\text{op}}\,|t|} || \ee^{\sum^N_{i=1} a_{N,i} u^i (Y-uYu^\dagger) (u^\dagger)^i t}-\ee^{P(Y-uYu^\dagger)t} ||_{\text{op}}
 \xrightarrow[N \to \infty]{} 0,
 \label{k42}
\end{eqnarray}
where we used the result from \eqref{b1} and $P(Y-uYu^\dagger)=0$. Substituting inequalities \eqref{k12}, \eqref{k22}, \eqref{k22g}, 
\eqref{k32}, and \eqref{k42} in \eqref{bvc2} we see that terms associated with $i=2$ are the slowest to converge. We have already shown using 
the properties of the matrix $(a)_{i,j}$ that
\begin{eqnarray}
 C^{(i,N-1)} \xrightarrow[N \to \infty]{} 0,  \nonumber \\
 \sum^{N-2}_{k=1} C^{(2,k)} \xrightarrow[N \to \infty]{} 0.  \nonumber 
\end{eqnarray}
We have also
\begin{eqnarray}
 &&a_{N,2} a_{N,1} + \sum^{N-1}_{i=2}a_{N,i+1} \left(a_{N,1}+\sum^{i-1}_{k=1} |a_{N,k+1}-a_{N,k}|+a_{N,i} \right) \leqslant \nonumber \\
 \leqslant && \left(a_{N,1}+\sum^{N-2}_{k=1} |a_{N,k+1}-a_{N,k}|+a_{N,N-1} \right) \sum^{N-1}_{i=1}a_{N,i+1} \leqslant \nonumber \\
 \leqslant && \left(a_{N,1}+\sum^{N-2}_{k=1} |a_{N,k+1}-a_{N,k}|+a_{N,N-1} \right) \xrightarrow[N \to \infty]{} 0.  \nonumber 
\end{eqnarray}
The only term left is $\sum^{N}_{i=1}a^2_{N,i}$ and we are going to argue in the following way: for every $N$ there exists a $k^*$ such
that for all $1\leqslant i \leqslant N$ 
\begin{equation}
 a_{N,i} \leqslant a_{N,k^*} \nonumber
\end{equation}
and thus
\begin{eqnarray}
 \sum^{N}_{i=1}a^2_{N,i} \leqslant \sum^{N}_{i=1}a_{N,i} a_{N,k^*} = a_{N,k^*} \xrightarrow[N \to \infty]{} 0, \nonumber 
\end{eqnarray}
because for all $i$ we have the relation $a_{N,i} \xrightarrow[N \to \infty]{} 0$. This results that for any element 
$X\in F_{{\cal B}(\mathcal{H})} \oplus \{Y-uYu^\dagger: Y \in {\cal B}(\mathcal{H})\}$
\begin{equation}
 ||\ee^{a_{N,1}uXu^\dagger t} \ee^{a_{N,2}u^2 X (u^\dagger)^2 t} \dots  \ee^{a_{N,N}u^N X (u^\dagger)^N t}-\ee^{P(X)t}||_{\text{op}} \xrightarrow[N \to \infty]{} 0.
\end{equation}

When $X \in F_{{\cal B}(\mathcal{H})} \oplus \overline{\{X-uXu^\dagger: X \in {\cal B}(\mathcal{H})\}}$ with $X=X_0+ W$  
and $W$ it is not of the form $X-uXu^\dagger$ then we reuse the strategy applied for the derivation of 
\eqref{specuse}.  We introduce for this case
\begin{eqnarray}
 P_0&=&\prod^{N}_{j=1} \ee^{a_{N,j} u^j X_1 (u^\dagger)^j t} \nonumber \\
 P_i&=&\prod^{N-i}_{j=1} \ee^{a_{N,j} u^j X_1 (u^\dagger)^j t}\prod^{N}_{j=N-i+1} \ee^{a_{N,j} u^j X_2 (u^\dagger)^j t},\quad 0 < i <N. \nonumber
\end{eqnarray}
Then
\begin{eqnarray}
 &&||\ee^{a_{N,1} u X_1 u^\dagger t} \dots  \ee^{a_{N,N} u^N X_1 (u^\dagger)^N t}-\ee^{a_{N,1} u X_2 u^\dagger t} \dots  
 \ee^{ a_{N,N} u^N X_2 (u^\dagger)^N t}||_{\text{op}} \leqslant
\nonumber \\
&& \leqslant \sum^{N-2}_{i=0}||P_i-P_{i+1}||_{\text{op}}+||P_{N-1}-\ee^{a_{N,1} uX_2u^\dagger t} \dots  
\ee^{ a_{N,N} u^N X_2 (u^\dagger)^N t}||_{\text{op}}. \nonumber
\end{eqnarray}
For $0 < i \leqslant N-1$,
\begin{eqnarray}
||P_{i-1}-P_i||_{\text{op}} \leqslant  \ee^{ \left(||X_1||_{\text{op}}+ ||X_1-X_2 ||_{\text{op}} \right)\,|t|} \,
|| \ee^{a_{N,i} X_1 t}-\ee^{a_{N,i} X_2 t}||_{\text{op}},\nonumber
\end{eqnarray}
where the derivation is done similarly as in Theorem \ref{theorem2}. We also have
\begin{eqnarray}
&&||P_{N-1}-\ee^{a_{N,1} u X_2 u^\dagger t} \dots  \ee^{a_{N,N} u^N X_2 (u^\dagger)^N t}||_{\text{op}}\leqslant \nonumber \\ 
 &&\leqslant \ee^{ \left(||X_1||_{\text{op}}+ ||X_1-X_2 ||_{\text{op}} \right)\,|t|} \,
|| \ee^{a_{N,N} X_1 t }-\ee^{a_{N,N} X_2 t}||_{\text{op}}.
\nonumber
\end{eqnarray}
Hence,
\begin{eqnarray}
 &&||\ee^{a_{N,1}uX_1u^\dagger t} \dots  \ee^{a_{N,N} u^N X_1 (u^\dagger)^N t}-\ee^{a_{N,1}uX_2u^\dagger t} \dots  
 \ee^{a_{N,N}u^N X_2 (u^\dagger)^N t}||_{\text{op}}
\nonumber \\
&& \leqslant \ee^{ \left(||X_1||_{\text{op}}+ ||X_1-X_2 ||_{\text{op}} \right)\,|t|} \,  
\sum^N_{i=1} || \ee^{a_{N,i} X_1 t}-\ee^{a_{N,i}X_2 t}||_{\text{op}}. \nonumber
\end{eqnarray}
We apply also inequality \eqref{specuse2} to have the following relation
\begin{eqnarray}
||\ee^{a_{N,i} X_1 t}-\ee^{a_{N,i}X_2t} ||_{\text{op}} &\leqslant&  ||X_1-X_2||_{\text{op}}
 a_{N,i} |t| \sum^\infty_{k=0} \frac{a^k_{N,i}\left(||X_1||_{\text{op}}+||X_2||_{\text{op}}\right)^k}{(k+1)!} |t|^k. \nonumber
\end{eqnarray}
Let $X_1=X_0+W$ where $W$ it is not of the form $X-uXu^\dagger$ and $ t \in \mathbb{C}$ with $|t|<\infty$. 
Then for any $\epsilon > 0$ we can take an arbitrary $\epsilon' > 0$ satisfying 
\begin{equation}
\ee^{ \left(||X_1||_{\text{op}}+ \epsilon' \right)\,|t|} |t| \epsilon' (1+ \epsilon')+ \epsilon' \leqslant \epsilon, \nonumber
\end{equation}
and for this $\epsilon'$ there exists $X_2 \in {\cal B}(\mathcal{H})$ and $Y \in {\cal B}(\mathcal{H})$ such that
$X_2=X_0+Y-uYu^\dagger$ and
\begin{equation}
 ||X_1-X_2||_{\text{op}}< \epsilon'. \nonumber
\end{equation}
Let us recall the result in Eq. \eqref{b1}, which means that for $\epsilon'$ we can take an $N'$ such that for every $N>N'$ 
\begin{equation}
||\ee^{a_{N,1} uX_2u^\dagger t} \ee^{a_{N,2} u^2 X_2 (u^\dagger)^2 t} \dots  \ee^{a_{N,N} 
u^N X_2 (u^\dagger)^N t}-
\ee^{P(X_1)t}||_{\text{op}}< \epsilon' \label{b2}
\end{equation}
where we used the relation $P(X_1)=P(X_2)=X_0$. On the other hand we have that $a_{N,i} \in [0,1)$ and 
$\lim_{N \to \infty} a_{N,i} =0$, which means that we can take an $N''$ such that for every $N>N''$
\begin{equation}
\sum^\infty_{k=1} \frac{a^k_{N,i} \left(||X_1||_{\text{op}}+||X_2||_{\text{op}} \right)^k}{(k+1)!}|t|^k < \epsilon'. \nonumber 
\end{equation}
This results that 
\begin{eqnarray}
  &&\ee^{ \left(||X_1||_{\text{op}}+ ||X_1-X_2 ||_{\text{op}} \right)\,|t|} \,  
\sum^N_{i=1} || \ee^{a_{N,i} X_1 t}-\ee^{a_{N,i}X_2 t}||_{\text{op}} < \nonumber \\
&&< \ee^{\left(||X_1||_{\text{op}}+\epsilon' \right)\,|t|} |t| \epsilon' (1+ \epsilon') \underbrace{\sum^N_{i=1} a_{N,i}}_{=1}. \label{b3}
\end{eqnarray}
Now, combining \eqref{b2} with \eqref{b3} we have that for all $N>\max\{N',N''\}$
\begin{eqnarray}
 &&||\ee^{a_{N,1}uX_1u^\dagger t} \ee^{a_{N,2}u^2 X_1 (u^\dagger)^2 t} \dots  \ee^{a_{N,N} u^N X_1 (u^\dagger)^N t}-\ee^{P(X_1)t}||_{\text{op}}  \nonumber \\
 && \leqslant ||\ee^{a_{N,1} uX_1u^\dagger t} \dots  \ee^{a_{N,N} u^N X_1 (u^\dagger)^N t}-\ee^{a_{N,1} u X_2 u^\dagger t} 
 \dots  \ee^{a_{N,N} u^N X_2 (u^\dagger)^N t}||_{\text{op}}+ \nonumber \\
 && + ||\ee^{a_{N,1} u X_2 u^\dagger t} \ee^{a_{N,2} u^2 X_2 (u^\dagger)^2 t} \dots  
 \ee^{a_{N,N} u^N X_2 (u^\dagger)^N t}-
 \ee^{P(X_1)t}||_{\text{op}} \nonumber \\
  && <  \ee^{\left(||X_1||_{\text{op}}+ \epsilon' \right)\,|t|} |t| \epsilon' (1+ \epsilon')+\epsilon' \leqslant \epsilon. \label{b4}
\end{eqnarray}
This shows that for all $X \in F_{{\cal B}(\mathcal{H})} \oplus \overline{\{X-uXu^\dagger: X \in {\cal B}(\mathcal{H})\}}$
\begin{eqnarray}
\lim_{N \to \infty} ||\ee^{a_{N,1}uXu^\dagger t} \ee^{a_{N,2}u^2 X (u^\dagger)^2 t} \dots  \ee^{a_{N,N}u^N X (u^\dagger)^N t}
-\ee^{P(X)t}||_{\text{op}}=0, \nonumber
\end{eqnarray}
and the proof is complete. 
\end{proof} 

This theorem yields that for certain splitting of the time interval and infinite number of pulses $N \to \infty$ the systems evolves like in the case of 
equidistant dynamical control. Thus, this results in many examples for the matrix $(a)_{i,j}$ satisfying the conditions of Theorem \ref{theorem4}, but adds
nothing to the question of optimising the convergence. Now, we shall develop a rather simple approach to this question. We consider only those elements $X$  which 
are in the set $\{X-uXu^\dagger: X \in {\cal B}(\mathcal{H})\}$ and inequality \eqref{optim} from Theorem \ref{theorem4}. In the first step we fix the value of $N$ and
assume for the sake of simplicity that $|t|=1/||Y||_{\text{op}}$ where $X=Y-u Y u^\dagger$. These simplifactions yield the following inequality
\begin{eqnarray}
&&||\ee^{a_{N,1}uXu^\dagger t} \ee^{a_{N,2}u^2 X (u^\dagger)^2 t} \dots  \ee^{a_{N,N}u^N X (u^\dagger)^N t}-1||_{\text{op}}\leqslant \nonumber \\
&& \sum^\infty_{i=2} \frac{2 \ee^{2}}{i!} \sum^{N-1}_{k=1} C^{(i,k)}+\sum^\infty_{n=1} \frac{\left(a_{N,1}+\sum^{N-1}_{i=1} |a_{N,i+1}-a_{N,i}|+a_{N,N}\right)^n}{n!}.
 \nonumber
\end{eqnarray}

In order to minimise the right hand side of the above inequality we make use of the following result.

\begin{proposition}
Assume $N$ real numbers with the following properties $a_{1}, a_{2}, \dots a_{N}\geqslant0$ and $\sum^N_{i=1} a_i =1$. Then
\begin{equation}
 \inf \left\{a_{1}+\sum^{N-1}_{i=1} |a_{i+1}-a_{i}|+a_{N} \right\}=\frac{2}{N} \nonumber
\end{equation}
if and only if $a_1=a_2=\dots=a_N=\frac{1}{N}$.
\end{proposition}
\begin{proof}
The proof in one direction is trivial. For the other direction, suppose first that $a_{N}\leqslant a_{N-1}$. We have the following inequality
\begin{equation}
 |x-a|+|x-b|\geqslant |b-a|,\,\, \text{for} \,\, a,b,x \in \mathbb{R}
 \label{inequse}
\end{equation}
and thus
\begin{equation}
 a_N+|a_N-a_{N-1}| \geqslant a_{N-1} \nonumber
\end{equation}
and we obtain its minimum when $a_N \in [0, a_{N-1}]$. Therefore a repeated application of \eqref{inequse} yields
\begin{equation}
a_{1}+\sum^{N-1}_{i=1} |a_{i+1}-a_{i}|+a_{N} \geqslant 2 a_1 \nonumber
\end{equation}
and the minimum is realised when 
\begin{equation}
a_N \leqslant a_{N-1} \leqslant \dots \leqslant a_1. \nonumber 
\end{equation}
Hence $N a_1 \geqslant \sum^N_{i=1} a_i =1$ and $a_1 \geqslant \frac{1}{N}$.  Consequently, 
\begin{equation}
a_{1}+\sum^{N-1}_{i=1} |a_{i+1}-a_{i}|+a_{N} \geqslant \frac{2}{N} \nonumber
\end{equation}
and the minimum is obtained when $a_1=a_2= \dots = a_N =\frac{1}{N}$. 

If $a_{N-1}\leqslant a_{N}$ then we have
\begin{equation}
a_{1}+\sum^{N-1}_{i=1} |a_{i+1}-a_{i}|+a_{N} \geqslant 2 a_N \nonumber
\end{equation}
and the minimum is realised when 
\begin{equation}
a_1 \leqslant a_{2} \leqslant \dots \leqslant a_N. \nonumber 
\end{equation}
Now, we have $ N a_N \geqslant \sum^N_{i=1} a_i =1$ and $a_N \geqslant \frac{1}{N}$.

Putting together both cases, we obtain the desired conclusion.
\end{proof}

It is worthwhile to note that the minimum of 
\begin{equation}
 a_{j}+\sum^{n-1}_{i=j} |a_{i+1}-a_{i}|+a_{n},\,\,\, j<n \nonumber
\end{equation}
is attained for $a_j=a_{j+1}= \dots =a_n$ which is a direct consequence of formula \eqref{inequse}. Therefore,
\begin{equation}
 \sum^N_{k=1} C^{(i,k)}+\left( a_{1}+\sum^{N-1}_{i=1} |a_{i+1}-a_{i}|+a_{N} \right)^i, \, \forall i \nonumber 
\end{equation}
is minimal when $a_1=a_2= \dots = a_N =\frac{1}{N}$. These results show that the inequality \eqref{optim} together with the definitions \eqref{Cels} and 
\eqref{Calt} derived for the proof of Theorem \ref{theorem4} takes its minimum when the weights of the time splitting are equal. 
This means that in the context of our derivation the case of equidistant dynamical control
is more favourable. 

\section{Conclusions}

In the context of dynamical control we have shown that the limit of infinite pulses for equidistant and non-equidistant single pulse "bang-bang" evolution
converges to the same limit. We have demonstrated that not only for unitary evolution but also for any uniformly continuous semigroup on a Hilbert space the generator
of the limit evolution is obtained by projecting the generator of the free evolution onto the commutant of unitary pulse. All the results are formulated for a linear
closed subspace of the Banach space of all bounded linear operators on a Hilbert space. This linear closed subspace is defined by the elements for which the  
$\lim_{N \to \infty} N^{-1} \sum^N_{i=1} T^i (X)$ exists for $T(X)=uXu^\dagger$ and $u$ being the unitary operator representing the pulse. The restriction is 
related to the fact that the map $X \rightarrow uXu^\dagger$ with $X$ being a bounded linear operator is not weakly compact. This issue falls into the more general problem of mean
ergodicity, whereas a power bounded map is called mean ergodic if the linear closed subspace defined by its Ces\`{a}ro mean is equal to the whole Banach space. If a Banach space is 
reflexive, then any power bounded map is mean ergodic on this Banach space \cite{Lorch}. For example, the map $T(X)=uXu^\dagger$ on the space of Hilbert-Schmidt operators is mean 
ergodic. In the case of finite dimensional Hilbert spaces all these complications and subtleties do not arise and the results obtained in this paper can be formulated for the 
whole space of square matrices.

We have derived inequalities for the proofs of the main theorems. In the case of non-equidistant dynamical control we have investigated 
the optimisation of these inequalities
for a fixed number of pulses. These inequalities are functions of the weights of the time splitting and we have found that the minimum is reached 
for the case of
equal weights. This means that in the context of our derivation the equidistant dynamical control is more effective than any 
other non-equidistant approach. 
However, in the case of specific systems it may happen that sharper inequalities can be derived, thus leading to a 
different result in the optimisation procedure.

This work is considered as a first step of proving convergence for a certain class of semigroup products defined by the problem of dynamical control. 
Thus this approach is open to further generalisations, for example considering contraction semigroups with unbounded generators.  

\ack
This work is supported by the Bundesministerium f\"ur Bildung und Forschung project Q.com. The author has profited from helpful discussions with A. B. Frigyik and F. Sokoli. In addition,
the author would like to thank the anonymous reviewer for her/his helpful and constructive comments that greatly contributed to 
improving the paper.

\end{document}